\documentclass{article}
\usepackage{amsfonts,amsmath,amsthm,amssymb,authblk,latexsym,bm,color,eucal}
\usepackage{hyperref} \hypersetup{colorlinks=true} 
\usepackage{citeref}
\usepackage{verbatim}

\newtheorem{theorem}{Theorem}[section]
\newtheorem{corollary}[theorem]{Corollary}

\newtheorem{lemma}[theorem]{Lemma}
\newtheorem{definition}[theorem]{Definition}
\newtheorem{example}[theorem]{Example}

\newtheorem{remark}[theorem]{Remark}

\numberwithin{equation}{section}

\def\wh{\widehat} \def\ol{\overline}
\def\+{\!+\!}

\begin{document}
\title{Self-dual 2-quasi Abelian Codes}
\author{
Liren Lin\par
{\small School of Optical Information and Energy Engineering}\par\vskip-1mm
{\small School of Mathematics and Physics}\par\vskip-1mm
{\small Wuhan Institute of Technology, Wuhan 430205, China}
\par\vskip2mm
Yun Fan\par
{\small School of Mathematics and Statistics}\par\vskip-1mm
{\small Central China Normal University, Wuhan 430079, China}
}
%\date{}
\maketitle

\insert\footins{\footnotesize{\it Email address}:
yfan@mail.ccnu.edu.cn (Yun Fan);
}

\begin{abstract}
A kind of self-dual quasi-abelian codes of index $2$ 
over any finite field~$F$ is introduced.  
By counting the number of such codes and the number of 
the codes of this kind whose relative minimum weights are small, 
such codes are proved to be asymptotically good 
provided $-1$ is a square in~$F$. 
Moreover, 
a kind of self-orthogonal quasi-abelian codes of index $2$
are defined; and such codes always exist. 
In a way similar to that for self-dual quasi-abelian codes of index $2$, 
it is proved that 
the kind of the self-orthogonal quasi-abelian codes of index $2$ 
is asymptotically good.

\end{abstract}

\section{Introduction}\label{introduction}

Let $F$ be a finite field with cardinality $|F|=q$
which is a power of a prime,
where $|S|$ denotes the cardinality of any set $S$. 
Let $n>1$ be an integer.
Any nonempty subset $C\subseteq F^n$ is called 
a code of length $n$ over $F$ in coding theory. 
The {\em Hamming weight} ${\rm w}(a)$ for 
any word $a=(a_1,\cdots,a_n)\in F^n$
is defined to be the number of the indexes $i$ that $a_i\ne 0$. 
The {\em Hamming distance} ${\rm d}(a,b)={\rm w}(a-b)$ for $a,b\in F^n$.
And ${\rm d}(C)=\min\{{\rm d}(c,c')\,|\,c\ne c'\in C\}$ 
is said to be the {\em minimum distance} of $C$,
while $\Delta(C)=\frac{{\rm d}(C)}{n}$ is called the
{\em relative minimum distance} of~$C$.
The rate of the code $C$ is defined as ${\rm R}(C)=\frac{\log_q|C|}{n}$.
A code sequence $C_1,C_2,\cdots$ is said to be
{\em asymptotically good} if the length $n_i$ of $C_i$ goes to infinity 
and, for $i=1,2,\cdots$, 
both the rate ${\rm R}(C_i)$ and the relative minimum distance
$\Delta(C_i)$ are positively bounded from below.
A class of codes is said to be {\em asymptotically good} if
there is an asymptotic good code sequence $C_1,C_2,\cdots$ in the class.

Let $G$ be an abelian group of order $n$.
By $FG$ we denote the group algebra, i.e., an $F$-vector space
with basis $G$ and with multiplication induced by the group multiplication of $G$.
Any ideal $C$ of $FG$
(i.e., any $FG$-submodule of the regular module $FG$) is
called an abelian code of length $n$ over $F$, or an $FG$-code for short.
Any element $a=\sum_{x\in G}a_xx\in FG$
 is identified with the word $(a_x)_{x\in G}\in F^n$.
Hence the Hamming weight ${\rm w}(a)$ of $a$ 
and the minimum Hamming weight ${\rm w}(C)$ of $C$ are defined.  
The euclidean inner product of $a=(a_x)$ and $b=(b_x)$ is defined to be 
$\big\langle a,b\big\rangle=\sum_{x\in G}a_xb_x$. 
Then the self-orthogonal codes and self-dual codes
are defined as usual, e.g., cf.~\cite{HP}.

Let $(FG)^2:=FG\times FG=\{(a,b)\,|\,a,b\in FG\}$, 
which is an $FG$-module. 
Any $FG$-submodule of $(FG)^2$ is called 
a {\em quasi-abelian code of index $2$} 
({\em $2$-quasi-abelian code} for short),
or a {\em quasi-$FG$ code of index $2$} ($2$-quasi-$FG$ code for short).

If $G$ is a cyclic group of order $n$, then
$FG$-codes and $2$-quasi-$FG$ codes  
 are the usual cyclic codes and 
$2$-quasi-cyclic codes, respectively. Note that 
%for the integer $n$, 
there is a unique cyclic group of order $n$ (upto isomorphism), 
but there may be many abelian groups of order $n$.

It is a long-standing open question whether or not   
the cyclic codes are asymptotically good, e.g., see \cite{MW06}.
But it has been  known for a long time that 
the $2$-quasi cyclic codes are asymptotically good, see \cite{CPW, C, K}.
Moreover, the self-dual quasi-cyclic codes with index going to infinity
are asymptotically good, see \cite{D, LS03}; 
and the binary (i.e., $q=2$) self-dual $2$-quasi-cyclic codes 
are asymptotically good, see \cite{MW}.

It is known that self-dual $2$-quasi-cyclic codes exist 
if and only if $-1$ is a square in $F$, see \cite{LS03}.
In fact, the condition ``$-1$ is a square in $F$'' is also necessary and sufficient
for the existence of the self-dual $2$-quasi-$FG$ codes for any abelian group 
$G$ of order $n$, see Corollary \ref{exist iff} below.

The residue integer ring ${\Bbb Z}_n$ modulo $n$ is partitioned into
$q$-cyclotomic cosets, see \cite[\S4.1]{HP}; 
and $\{0\}$ is obvious a $q$-cyclotomic coset 
which we call the {\em trivial $q$-cyclotomic coset}.  
By $\mu_q(n)$ we denote the minimal size 
of the non-trivial $q$-cyclotomic cosets of ${\Bbb Z}_n$. 

By assuming that Artin's primitive root conjecture holds, i.e., 
for a non-square $q$, there exist infinitely many odd primes $p_1,p_2,\cdots$ satisfying:
\begin{itemize}
\item\vskip-5pt
$\mu_q(p_i)=p_i-1$, $i=1,2,\cdots$, 
\end{itemize} \vskip-5pt
Alahmadi,  \"Ozdemir and Sol\'e \cite{AOS} proved that, 
if $-1$ is a square in $F$ and $q$ is a non-square, 
then there is an asymptotically good
sequence of self-dual $2$-quasi-cyclic codes $C_1,C_2,\cdots$
with code length of $C_i$ equal to $2p_i$ for $i=1,2,\cdots$. 
And they asked an open problem if the dependence on 
Artin's primitive root conjecture can be removed.

In the dissertation \cite{L PhD}, 
it has been proved by a probabilistic method that, 
if $-1$ is a square in $F$, and
there are odd integers $n_1,n_2,\cdots$ coprime to $q$ 
satisfying the following two:
\begin{itemize} 
\item \vskip-5pt
the multiplicative order of $q$ modulo ${n_i}$ is odd for $i=1,2,\cdots$, 
\item \vskip-5pt
$\lim\limits_{i\to\infty}\frac{\log_q n_i}{\mu_q(n_i)}=0$; 
\end{itemize}
then for any abelian groups $G_i$ of order $n_i$, $i=1,2,\cdots$, 
there exists an asymptotically good sequence $C_1,C_2,\cdots$
with $C_i$ being self-dual $2$-quasi-$FG_i$ codes for $i=1,2,\cdots$.

It was known for a long time that there exist % The existence of 
odd integers $n_1,n_2,\cdots$ coprime to $q$ satisfying that
$\lim\limits_{i\to\infty}\frac{\log_q n_i}{\mu_q(n_i)}=0$, 
see \cite[Lemma 2.6]{BM}, \cite[Lemma II.6]{FL20} 
or Lemma \ref{n_i} (and Remark \ref{r n_i}) below. 
In fact,  this existence is just enough to guarantee the asymptotic goodness
of the self-orthogonal $2$-quasi-abelian codes,   
and, to guarantee the asymptotic goodness of the self-dual $2$-quasi-abelian codes
once $-1$ is a square in $F$. 
More precisely,  we have the following. 

\begin{theorem}\label{main result}
Let $n_1,n_2,\cdots$ be odd positive integers coprime to $q$ satisfying that 
$\lim\limits_{i\to\infty}\frac{\log_q n_i}{\mu_q(n_i)}=0$. 
Let $G_i$ be any abelian group of order $n_i$ for $i=1,2,\cdots$.

{\rm(1)} 
There exist self-orthogonal $2$-quasi-$FG_i$ codes $C_i$ of dimension $n_i-1$,
$i=1,2,\cdots$, 
such that the code sequence $C_1,C_2,\cdots$ is asymptotically good.

{\rm(2)} If $-1$ is a square in $F$, then
there exist self-dual $2$-quasi-$FG_i$ codes $C_i$, $i=1,2,\cdots$, 
such that the code sequence $C_1,C_2,\cdots$ is asymptotically good.
\end{theorem}

We will prove the theorem by counting the number of 
the self-dual (or self-orthogonal) $2$-quasi-abelian codes we considered.

In Section \ref{abelian} we describe fundamentals of abelian group algebras. 

In Section \ref{2-quasi}, we investigate the structure of 
$2$-quasi abelian codes, and introduce the $2$-quasi-abelian codes of type I.
If  the abelian group $G$ is cyclic,
then the $2$-quasi-$FG$ codes of type I we defined 
are just the so-called double circulant codes in literature, 
e.g., see \cite{AOS, TTHAA}.

Section \ref{self-dual type I} is devoted to the study on
the self-dual $2$-quasi-abelian codes of type~I. 
First we show a counting formula of the number of such codes.
Then we exhibit an estimation of the number 
of the self-dual $2$-quasi-abelian codes of type I whose relative minimum
weights are small. Finally, we prove the above Theorem \ref{main result}(2),
see Theorem \ref{D asymptotically good}.  

As mentioned above, 
the existence of self-dual $2$-quasi-abelian codes is conditional.
In Section \ref{self-orth type I^dag} we discuss 
a kind of self-orthogonal $2$-quasi-$FG$ codes of dimension $n-1$, 
where $G$ is any abelian group of order $n$ as before.  
The existence of such codes is unconditional. 
The study method for them is similar to that in Section \ref{self-dual type I}.
In Theorem \ref{D^dag asymptotically good} we 
complete the proof of Theorem \ref{main result}(1).

\section{Abelian codes}\label{abelian}

From now on we always assume that
 $F$ is a finite field of cardinality $|F|=q$, 
 $n>1$ is an odd integer with $\gcd(n,q)=1$, 
 and $G$ is an abelian group of order~$n$.

Let $FG=\left\{a=\sum_{x\in G}a_xx\,|\,a_x\in F\right\}$ 
be the group algebra of $G$ over $F$. 
As mentioned in Section \ref{introduction}, 
%i.e., $FG$ is an $F$-vector space with basis $G$, 
%and the multiplication of $FG$ is 
%induced by the multiplication of the group $G$. 
any $\sum_{x\in G}a_x x\in FG$ is viewed 
as a sequence $(a_x)_{x\in G}$ of $F$ indexed by $G$, 
and the Hamming weight, the euclidean inner product etc.
are defined as usual. 

The map $G\to G$, $x\mapsto x^{-1}$, 
is an automorphism of $G$ of order $2$
(since the order of $G$ is odd).
And the following map
\begin{equation}\label{bar map}
\textstyle
%{\rm bar}: 
FG\to FG,~~ a\mapsto\ol a,
\quad\mbox{where $\ol a=\sum_{x\in G}a_x x^{-1}$  for $a=\sum_{x\in G}a_x x$,} 
\end{equation}
is an algebra automorphism of $FG$ of order $2$. 
Following \cite{FL20}, 
we call the automorphism~Eq.\eqref{bar map} the {\em ``bar'' map} of $FG$.
The following is clearly a linear form on $FG$:
\begin{equation}
\textstyle
\sigma: FG\to F,~~ \sum_{x\in G}a_x x\mapsto a_{1_G},
\end{equation}
where $a_{1_G}$ is the coefficient of the identity element $1_G$ in the
linear combination $a=\sum_{x\in G}a_x x$.
Recall that the euclidean inner product
$\langle a,b\rangle=\sum_{x\in G}a_x b_x$ 
for $a=\sum_{x\in G}a_x x$ and $b=\sum_{x\in G}b_x x$.
Then the following holds (\cite[Lemma II.4]{FL20}):
\begin{equation}\label{inner}
\langle a,b\rangle=\sigma(a\ol b)
=\sigma(\ol a b)=\sigma(b\ol a)=\sigma(\ol b a), 
~~~\forall~ a,b\in FG.
\end{equation}

Since $\gcd(n,q)=1$, by Maschke's Theorem (see \cite[pp. 116-117]{AB}), 
$FG$ is a semisimple algebra, i.e.,
$FG$ is the direct sum of simple ideals as follows: 
\begin{equation*}
 FG=A_0\oplus A_1\oplus\cdots\oplus A_m.
\end{equation*}
Correspondingly, the identity element $1=1_G$ of the algebra 
has a unique decomposition
$$1=e_0+e_1+\cdots+e_m,$$
with $e_i\in A_i$, $0\le i\le m$.
It is easy to check that
$$
e_ie_j=\begin{cases}e_i, & i=j;\\ 0,& i\ne j; \end{cases}~~\forall~0\le i,j\le m.
$$
Such $e_0,e_1,\cdots,e_m$ are called ``orthogonal'' idempotents
(here ``orthogonal'' is different from the orthogonality defined by inner product).
An idempotent $0\ne e=e^2$ is said to be primitive if 
it is impossible to write $e=e'+e''$ for two
non-zero idempotents $e'$ and $e''$ with $e'e''=0$.
So we have: 
\begin{itemize}
\item 
For $i=0,1,\cdots,m$,  
$A_i=FG e_i$ is a simple ring with identity $e_i$, hence 
$e_i$ is a primitive idempotent and $A_i$
is a field extension over $F$.
We can assume that $A_0=FGe_0=Fe_0$ is the trivial $FG$-module of dimension~$1$.
\item
$E=\{e_0\!=\!\frac{1}{n}\sum_{x\in G} x, ~ e_1,\cdots,e_m\}$ 
is the set of all primitive idempotents of~$FG$, and $1=\sum_{e\in E}e$.
\item
The ``bar'' map Eq.\eqref{bar map} permutes the primitive
idempotents $e_0,e_1,\cdots,e_m$ (equivalently, the ``bar'' map permutes
the simple ideals $A_0,A_1,\cdots,A_m$); obviously, $\ol e_0=e_0$.
\end{itemize}
Then we can write the set $E$ of all primitive idempotents as 
%$E=\{e_0\}\cup E'\cup E''$, where 
%\begin{equation}\label{E' E''}
%\begin{array}{l}
%E'=\{e_1,\cdots,e_r \},
% \quad \ol e_i=e_i \mbox{~ for ~} i=1,\cdots,r;
%\\[4pt]
%E''=\{e_{r\+ 1},\ol e_{r\+ 1},\,\cdots,\, e_{r\+ s},\ol e_{r\+ s}\}, 
%  \quad \ol e_{r\+ j}\ne e_{r\+j} \mbox{~ for ~}  j=1,\cdots,s.
%\end{array}
%\end{equation}
\begin{equation}\label{E}
E=\{e_0\}\cup \{e_1,\cdots,e_r\}\cup\{e_{r\+ 1},\ol e_{r\+ 1},\,\cdots,\, e_{r\+ s},\ol e_{r\+ s}\},
\end{equation}
where 
$\ol e_i=e_i$  for  $i=1,\cdots,r$,
$\ol e_{r\+ j}\ne e_{r\+j}$  for $ j=1,\cdots,s$ 
and $1+r+2s=m$. 
%For $i=1,\cdots,r$, $\ol A_i=A_i$.
%By \cite[Lemma 3.3(2)]{FL20},
%$\dim A_i=FGe_i$ is even.
For $j=1,\cdots,s$, 
the  restriction of the ``bar'' map to $A_{r\+j}$ induces a map:
$$ A_{r\+j}\!=\!FG e_{r\+j}~\longrightarrow \ol A_{r\+j}\!=\! FG\ol e_{r\+j},~~~ 
   a \longmapsto\ol a,$$ 
which is an $F$-algebra isomorphism.
%; in particular 
%$\dim A_{r\+ j}=\dim \ol A_{r\+ j}$.
It is easy to check that
$A_{r\+ j}+\ol A_{r\+ j}$ 
are invariant by the automorphism ``bar''. 
So we set
\begin{equation}\label{d hat E}
\begin{array}{l}
 \wh e_{r\+j}=e_{r\+j}+\ol e_{r\+j},~~~
 \wh A_{r\+j}=FG\wh e_{r\+j}=A_{r\+j}+\ol A_{r\+j},\quad j=1,\cdots,s;
 \\[5pt]
 \wh E=\{e_0,\,e_1,\cdots,e_r,\, \wh e_{r\+1},\cdots,\wh e_{r\+s}\}.
\end{array}
\end{equation}
Then
\begin{align}\label{e hat E}
 1=\sum_{e\in\wh E}e;~~~~~ 
 \ol e=e, ~\forall~e\in\wh E;~~~~~
 ee'=\begin{cases}e, \!& e=e';\\ 0, \!& e\ne e'; \end{cases}~\forall~e,e'\in\wh E;
\end{align}
and we have a decomposition of $FG$ written in several ways:
\begin{equation}\label{FG=}\textstyle
\begin{array}{l}
 FG
=FGe_0\bigoplus\big(\bigoplus_{i=1}^r FG e_i\big) \bigoplus
     \big(\bigoplus_{j=1}^s \left( FGe_{r\+j}\bigoplus FG\ol e_{r\+j}\right)\big)\\[5pt]
    \hskip7mm =~~A_0\;\bigoplus\;\big(\bigoplus_{i=1}^r A_i\big) \; \bigoplus \;
             \big(\bigoplus_{j=1}^s \left( A_{r\+j}\bigoplus \ol A_{r\+j}\right)\big)\\[5pt]
\hskip7mm=~~A_0\;\bigoplus\;\big(\bigoplus_{i=1}^r A_i\big) \; \bigoplus \;
             \big(\bigoplus_{j=1}^s\wh A_{r\+j}\big)\\[5pt]
  \hskip7mm=~~FGe_0\bigoplus\big(\bigoplus_{i=1}^r FG e_i\big) \bigoplus
         \big(\bigoplus_{j=1}^s FG\wh e_{r\+j}\big).
\end{array}
\end{equation}

For an ideal $A$ of $FG$,
%which is invariant by the ``bar'' map,
 we denote $A^\flat=\{a\mid a\in A, \ol a=a\}$.

\begin{lemma}\label{FG^flat} 
Keep the notation as above.

{\rm(1)} $A_0^\flat=A_0=FGe_0=Fe_0$.

{\rm(2)}  $(FG)^{\flat}$ is a subalgebra of $FG$ as follows:
\begin{equation}\label{FG^dag=}\textstyle
  (FG)^{\flat}=Fe_0\bigoplus\left(\bigoplus_{i=1}^r A_i^\flat\right) \bigoplus
 \left(\bigoplus_{j=1}^s(\wh A_{r\+j})^\flat\right);
% \left(\bigoplus_{j=1}^s(C_{r\+j}+\ol C_{r\+j})^\flat\right);
\end{equation}
and $\dim (FG)^\flat=1+\frac{n-1}{2}$.

{\rm(3)} For $j =1,\cdots,s$, 
$$(\wh A_{r\+j})^\flat=(A_{r\+j}+\ol A_{r\+j})^\flat=\{a+\ol a\mid a\in A_{r\+j}\};$$
in particular, 
$\dim(\wh A_{r\+j})^\flat=\dim(A_{r\+j}+\ol A_{r\+j})^\flat=\dim A_{r\+j}
=\dim\ol A_{r\+j}$.

{\rm(4)} For $i=1,\cdots,r$,  $\dim A_i$ is even, and 
$\dim A_i^\flat=\frac{1}{2}\dim A_i$.
\end{lemma} 

\begin{proof}
(1). Obvious.

(2). 
For any $a\in FG$, by Eq.\eqref{FG=},
we assume $a=\sum_{e\in\wh E}a_e$ with $a_e\in FGe$ for $e\in \wh E$.
Then 
$\ol a=\sum_{e\in\wh E}\ol a_e$.
For any $e\in \wh E$,
since $\ol e=e$ and $\ol{FGe}=FG\ol e=FGe$, 
we have $\ol a_e\in FGe$,
then $a=\ol a$ iff $a_e=\ol a_e$ for $e\in \wh E$,
i.e., Eq.\eqref{FG^dag=} holds.

For any $1\ne x\in G$, $\ol x=x^{-1}\ne x$ (because the order $n$ of $G$ is odd). 
Thus, $a=\sum_{x\in G}a_xx \in (FG)^{\flat}$ if and only if
$a_{x^{-1}}=a_x$ for all $1\ne x\in G$. 
In other words, except for $1\in G$, there are 
$\frac{n-1}{2}$ coefficients of $a\in (FG)^{\flat}$ 
which can be chosen from $F$ freely.
That is, $\dim (FG)^{\flat}=1+\frac{n-1}{2}$.

(3). For $a\in A_{r\+j}$ and $b\in \ol A_{r\+j}$,
since $\ol a\in \ol A_{r\+j}$ and $\ol b\in A_{r\+j}$,
then
$$
 a+b\in (A_{r\+j}+\ol A_{r\+j})^\flat\iff
 \ol a+\ol b=a+b \iff \ol a=b~\mbox{and}~\ol b=a. 
$$

(4).
For $i=1,\cdots,r$,
since $A_i$ is a finite field and the ``bar'' map is an automorphism of it
of order $\le 2$,  
then we have
$$
 \dim A_i^\flat=
 \begin{cases}
  \dim A_i, & \mbox{``bar'' is the identity automorphism of $FGe_i$;}\\
  \frac{1}{2}\dim A_i, & \mbox{otherwise.}\\
 \end{cases}
$$
In particular, 
$\frac{1}{2} \dim A_i\le \dim A_i^\flat$.
By (3), 
$$\textstyle\dim (\wh A_{r\+j})^\flat=\dim A_{r\+j}=\frac{1}{2}\dim (A_{r\+j}+\ol A_{r\+j}).$$
By Eq.\eqref{FG^dag=},
\begin{align*}
\textstyle 1+ \frac{n-1}{2}
 &\textstyle =1+\sum_{i=1}^r\dim A_i^\flat
    +\sum_{j=1}^s\dim (\wh A_{r\+j})^\flat\\
 &\textstyle \ge 1+\frac{1}{2}\sum_{i=1}^r\dim A_i 
 +\frac{1}{2}\sum_{j=1}^s(\dim A_{r\+j}+\dim\ol A_{r\+j})
 =1+\frac{n-1}{2}.
\end{align*}
So, for $i=1,\cdots,r$, every %inequality 
``$\frac{1}{2} \dim A_i\le \dim A_i^\flat$''  
has to be an equality; that is,
$\frac{1}{2} \dim A_i=\dim A_i^\flat$, for $i=1,\cdots,r$.
\end{proof}

\medskip
Set
\begin{equation}\label{hat E^dag}
\wh{E}^\dag=\wh E-\{e_0\}=
\{e_1,\cdots,e_r, \,\wh e_{r\+1},\cdots, \wh e_{r\+s}\}.
\end{equation}
By Lemma \ref{FG^flat},
for any $e\in\wh E^\dag$,
$\dim FGe$ is even.
Denote 
\begin{align*}
&\textstyle n_e=\frac{1}{2}\dim FGe,~~e\in \wh E^\dag; \\
& n_i=n_{e_i}, ~ i=1,\cdots,r; ~~~ 
 n_{r+j}=n_{\wh e_{r+j}}, ~ j=1,\cdots,s;
\end{align*}
then 
\begin{equation}\label{dim n-1}\textstyle
 n=1+2\sum_{e\in \wh E^\dag}n_e
 = 1+ 2\big(\sum_{i=1}^r n_i+\sum_{j=1}^s n_{r\+j}\big).
\end{equation}

\begin{remark}\rm\label{E_C}
For any ideal  $A\le FG$,
 there is a unique subset $E_A\subseteq E$
such that 
$$\textstyle
 A=\bigoplus_{e\in E_A}FGe=FGe_A,\quad   
 \mbox{where~ } e_A=\sum_{e\in E_A} e.
$$ 
%and vice versa. 
So any ideal $A$ is a ring with identity $e_A$ and with  the unit group
$$\textstyle
 A^\times=\prod_{e\in E_A}(FGe)^\times,
$$
where $(FGe)^\times=FGe-\{0\}$
since  $FGe$ is a finite field. 
Thus, for $a\in A$, $FGa=A$ if and only if $a\in A^\times$, 
if and only if $ae\ne 0$ for all $e\in E_A$.

In particular, the multiplicative unit group of $FG$ is as follows: 
\begin{equation*}\label{FG^times}
(FG)^\times=F^\times e_0\mbox{\,\large$\times$\,}
\Big( \mathop{\mbox{\large$\times$}}\limits_{i=1}^r A_i^\times\Big)
\mbox{\,\large$\times$\,} 
\Big(\mathop{\mbox{\large$\times$}}\limits_{j=1}^s 
 \big(A_{r\+j}^\times\mbox{\large$\times$}\ol A_{r\+j}^\times\big)\Big).
\end{equation*}
\end{remark}

\section{$2$-quasi-abelian codes}\label{2-quasi}

Keep the notation in Section \ref{abelian}.

The outer direct sum 
$(FG)^2:=FG\times FG=\{(a,b)\,|\,a,b\in FG\}$ is an $FG$-module.
Any $FG$-submodule $C$ of $(FG)^2$, denoted by $C\le(FG)^2$, 
 is said to be a {\em $2$-quasi-abelian code},
or a {\em $2$-quasi-$FG$ code}. 
Denote by ${\rm w}(a,b)$ the Hamming weight of $(a,b)\in (FG)^2$,
and by  ${\rm w}(C)$ the minimum Hamming weight  of $C$. 
Note that $C$ is a linear code of length $2n$.
The relative minimum distance $\Delta(C)=\frac{{\rm w}(C)}{2n}$, 
and the rate $R(C)=\frac{\dim C}{2n}$. 
%A class of codes is said to be {\em asymptotically good} if
%there is a code sequence $C_1,C_2,\cdots$ in the class such that
%the length $n_i$ of $C_i$ goes to infinity and both
%the rate ${\rm R}(C_i)$ and the relative minimum distance
%$\Delta(C_i)$ are positively bounded from below.

%Recall from Eq.\eqref{E' E''} and Eq.\eqref{hat e_r+j} that 
%\begin{equation}\label{d hat E}
%\begin{array}{l}
%1=e_0+e_1+\cdots+e_r+\wh e_{r\+1}+\cdots+\wh e_{r\+s},\\[4pt]
%FG=Fe_0\oplus FGe_1\oplus\cdots\oplus FGe_r  
% \oplus FG\wh e_{r\+1}\oplus\cdots\oplus FG\wh e_{r\+s},\\[4pt]
 %\dim FGe_1=2n_1, ~\cdots, ~ \dim FGe_r=2n_r, \\[4pt]
 %\dim FG\wh e_{r\+1}=2n_{r\+1},~\cdots,~ \dim FG\wh e_{r\+s}=2n_{r\+s};
%\end{array}
%\end{equation}
%hence $1+2\sum_{k=1}^{r+s} n_k=n$.
\medskip
For any $(a,b)\in (FG)^2$,
by Eq.\eqref{e hat E} we have 
$$(a,b)=1\cdot (a,b)= e_0(a,b) + e_1(a,b)+\cdots+ e_r(a,b)  
 + \wh e_{r\+1}(a,b)+\cdots+ \wh e_{r\+s}(a,b).$$
 Thus for any $C\le (FG)^2$,
\begin{equation}\label{C=...+...}
C=e_0 C\oplus e_1C\oplus\cdots\oplus e_r C 
 \oplus \wh e_{r\+1}C\oplus\cdots\oplus \wh e_{r\+s}C.
\end{equation}

\begin{lemma}\label{(FG)^2=...oplus...}
{\rm(1)}
As $FG$-modules we have a direct sum decomposition: 
$$(FG)^2=(Fe_0)^2\oplus(FGe_1)^2\oplus\cdots\oplus (FGe_r)^2
 \oplus (FG\wh e_{r\+1})^2\oplus\cdots\oplus (FG\wh e_{r\+s})^2.
$$

{\rm(2)}
For any $e\ne e'\in \wh E
 =\{e_0,e_1,\cdots,e_r, \wh e_{r\+1},\cdots, \wh e_{r\+s}\}$, 
 the submodules $(FG e)^2$ and $(FG e')^2$ of $(FG)^2$ are orthogonal, i.e., 
$\big\langle (FG e)^2,\,(FG e')^2\big\rangle=0$.

{\rm(3)}
For any $e\in\wh E$, %in $(FG)^2$ 
the orthogonal submodule 
$((F\!Ge)^2)^\bot\!=\!\bigoplus_{e'\in\wh E-\{e\}}\!(F\!G e')^2$.
\end{lemma}

\begin{proof}
(1).
This is checked by Eq.\eqref{C=...+...} and the following: 
for $e\in \wh E$,  
$$e(FG)^2=\{(ea,eb)\,|\,a,b\in FG\}=FGe\times FGe=(FGe)^2.$$

(2). 
Note that $\ol {e'}=e'$ and $ee'=0$. 
For $(ae,be)\in(FGe)^2$ and $(a'e',b'e')\in(FGe')^2$, where $a,b,a',b'\in FG$, 
by Eq.\eqref{inner} we have 
\begin{align*}
&\big\langle(ae,be),\,(a'e',b'e')\big\rangle
 =\sigma(ae\,\ol{a'e'})+\sigma(be\,\ol{b'e'})\\
 &  =\sigma(ae\ol {e'}~\ol {a'})+\sigma(be\ol {e'}~\ol {b'})
 =\sigma(aee'\ol {a'})+\sigma(bee'\ol{b'})=0.
\end{align*}

(3). This follows from (2) immediately. 
\end{proof}

Note that, for $e\in\wh E$, $FGe$ is an algebra (with identity $e$)
and $(FGe)^2$ is an $FGe$-module, 
and Lemma \ref{(FG)^2=...oplus...} %(2) and (3) 
implies that
the restriction to $(FGe)^2$ of the inner product of $(FG)^2$ 
is non-degenerate.

For the decomposition of $C\le(FG)^2$ in Eq.\eqref{C=...+...},
the component $eC$ for $e\in\wh E$ is an $FGe$-submodule of $(FGe)^2$.

\begin{corollary}\label{C orth} 
Let $C\le(FG)^2$ be as in Eq.\eqref{C=...+...}. 
Then the orthogonal submodule $C^\bot$ has a direct decomposition: 
$$
C^\bot=(e_0C)^{\bot_0}\oplus(e_1C)^{\bot_1}\oplus\cdots\oplus(e_rC)^{\bot_r}
 \oplus(\wh e_{r\+1}C)^{\bot_{r\+1}}\oplus\cdots\oplus(\wh e_{r\+s}C)^{\bot_{r\+s}},
$$
where $(e_iC)^{\bot_i}$ for $i=0,1,\cdots,r$ denotes
the orthogonal $FGe_i$-submodule in $(FGe_i)^2$, and
$(\wh e_{r\+j}C)^{\bot_{r\+j}}$ for $j=1,\cdots,s$ denotes
the orthogonal $FG\wh e_{r\+j}$-submodule in $(FG\wh e_{r\+j})^2$. 
\end{corollary}

\begin{proof}
By Lemma \ref{(FG)^2=...oplus...}(2), the right hand side of the wanted equality
is contained in the left hand side. 
By computing the dimensions of the two sides of the wanted equality,
we obtain that the wanted equality does hold. 
\end{proof}

\begin{corollary} \label{C self-dual} 
Let $C\le(FG)^2$ be as above. Then $C$ is self-orthogonal if and only if
$eC$ is self-orthogonal in $(FGe)^2$ for any $e\in\wh E$.
\end{corollary}

%Keep the notation introduced in Section \ref{abelian}. In this section
%we introduce a kind of self-dual $2$-quasi-abelian codes.

\medskip
In general, a $2$-quasi-$FG$ code may not be  generated by one element.
For our purpose, in the following we consider the $2$-quasi-$FG$ codes
 generated by one element. 
 For any $(a,b)\in(FG)^2$,
  let 
  \begin{equation}\label{one gen}
  C_{a,b}=\{(ua,ub)\,|\,u\in FG\}
  \end{equation}
be the $FG$-submodule of $(FG)^2$ generated by $(a,b)$.

\begin{lemma}\label{dimension of Ca,b}
As $FG$-modules,
$C_{a,b}\cong FG/{\rm Ann}_{FG}(a,b)$,
where ${\rm Ann}_{FG}(a,b)=\{u\,|\, u\in FG, ua=0=ub\}$
denotes the annihilator of $\{a,b\}$ in $FG$.
In particular, $\dim C_{a,b}\le n$.
\end{lemma}

\begin{proof}
By the construction of $C_{a,b}$,
there is a surjective $FG$-module homomorphism 
from $FG$ to $C_{a,b}$,
mapping $u$ to $(ua,ub)$ for any $u\in FG$.
It is easy to check that the kernel is just ${\rm Ann}_{FG}(a,b)$.
Then we have 
$FG/{\rm Ann}_{FG}(a,b)\cong C_{a,b}$.
\end{proof}

If one of $a$ and $b$ is invertible, then
${\rm Ann}_{FG}(a,b)=0$,
hence $\dim C_{a,b}=n$ obviously. 
Without loss of generality, we assume that $a\in(FG)^\times$ is a unit, 
%then $(1, a^{-1}b)=(a^{-1}a,a^{-1}b)\in C_{a,b}$ and
%$C_{1, a^{-1}b}=C_{a,b}$. 
then $(a,b)=a(1, a^{-1}b)\in C_{1,a^{-1}b}$ and
$C_{a,b}=C_{1, a^{-1}b}$.

\begin{definition}\label{d type I}\rm
For any $b\in FG$, we call $C_{1,b}=\{(u,ub)\,|\,u\in FG\}\le(FG)^2$
a $2$-quasi-$FG$ code {\em of type I}, or a $2$-quasi-abelian code of {\em type I}.
\end{definition}

\begin{remark}\label{r type I}\rm
List the elements of $G$ as 
\begin{equation}\label{list G} 
G=\{x_0=1,x_1,\cdots,x_{n-1}\}. 
\end{equation}
Any $a=\sum_{j=0}^{n-1}a_jx_j\in FG$ corresponds to a sequence $(a_0,a_1,\cdots,a_{n-1})\in F^n$. Thus, 
as a linear code, $C_{1,b}$ is an $F$-subspace of $(FG)^2$ with basis as follows:
\begin{equation}\label{1,b}
 (1,b),~(x_1, x_1b), ~\cdots,~(x_{n-1},x_{n-1}b).
\end{equation}
Of course, $1$ corresponds to the sequence $(1,0,\cdots,0)\in F^n$. 
Set $b=\sum_{j=0}^{n-1}b_jx_j$, i.e., $b$ corresponds to the sequence 
$(b_0, b_1,\cdots,b_{n-1})\in F^n$.
Any $x_i\in G$ provides a Cayley permutation $\rho_i$: 
\begin{equation}\label{rho_i}
 x_i \longmapsto \rho_i\!=\!\begin{pmatrix}
 0 & 1 & \cdots & j & \cdots & n-1\\
 0' & 1' & \cdots & j' & \cdots & (n-1)'
\end{pmatrix},
%\mbox{where $x_ix_j=x_{j'}$ for $j=0,1,\cdots,n-1$}
\end{equation}
where $x_{j'}=x_ix_j$ for $j=0,1,\cdots,n-1$; in particular,  $0'=i$ because $x_i x_0=x_i$.
The map $x_i\mapsto\rho_i$ is said to be 
the {\em Cayley representation} of the group $G$, 
e.g., see~\cite[pp. 28]{AB}.
%We call $\rho_i$ the {\em Cayley permutation} of $x_i$. 
Clearly, $\rho_0$ is the identity permutation. Then 
$$\textstyle
 x_ib=\sum_{j=0}^{n-1}b_jx_ix_j=\sum_{j=0}^{n-1}b_jx_{j'};
$$ 
i.e.,
  $x_ib$ corresponds to the sequence obtained 
by $\rho_i$-{permuting} on the sequence $(b_0, b_1,\cdots,b_{n-1})$. 
Let $B$ be the $n\times n$ matrix whose $i$'th row is obtained
by $\rho_i$-permuting on $(b_0, b_1,\cdots,b_{n-1})\in F^n$. 
By $I$ we denote the identity matrix.
\end{remark}

\begin{lemma}\label{l type I} 
Keep the notation in Remark \ref{r type I}.
%Eq.\eqref{list G} and Eq.\eqref{rho_i}.
Then a linear code $C$ of length $2n$ and dimension $n$
has a generating matrix $(I~B)$ if and only if $C=C_{1,b}$. 
%where $b=\sum_{j=0}^{n-1}b_jx_j\in FG$.
\end{lemma}

\begin{proof}
Assume that $C=C_{1,b}$. By the above analysis, $(x_i,x_ib)$ in Eq.\eqref{1,b}
corresponds to the $i$'th row of the matrix $(I~B)$.
So $C$ is a $[2n,n]$ linear code having $(I~B)$ as a generating matrix.

Conversely, assume that $C$ is a $[2n,n]$ linear code 
having $(I~B)$ as a generating matrix.
The $i$'th row of the matrix $(I~B)$ corresponds to the element $(x_i,x_ib)\in(FG)^2$.
Thus $C$ can be regarded as an $F$-subspace of $(FG)^2$
with basis Eq.\eqref{1,b}; hence $C=C_{1,b}$.
\end{proof}

\begin{example}\label{double circulant}\rm
A typical example is to take $G=\{1,x,x^2,\cdots,x^{n-1}\}$ 
to be a cyclic group of order $n$. Then 
$\rho_1$ is just the cyclic permutation, $\rho_i=\rho_1^i$, 
and $B$ is just the usual circulant matrix with first row $(b_0,b_1,\cdots,b_{n-1})$. 
The code $C_{1,b}$ with generating matrix $(I~B)$ is said to be a
{\em double circulant code} in literature, e.g., see \cite{AOS, TTHAA}. 
\end{example}

We consider self-dual $2$-quasi-$FG$ codes of type I.

\begin{lemma}\label{orth}
For any $C_{a,b}\le(FG)^2$ as in Eq.\eqref{one gen},  
the following three are equivalent to each other: 
\begin{itemize}
\item[\rm(1)] \vskip-5pt
 $C_{a,b}$ is self-orthogonal;
\item[\rm(2)] \vskip-3pt
 $a\ol a+b\ol b=0$;
\item[\rm(3)]  \vskip-3pt
the following two hold:
\begin{itemize}
 \item[\rm(3.i)] \vskip-4pt 
  $ae_i\ol{ae_i}+be_i\ol{be_i}=0$, for $i=0,1,\cdots,r$; 
 \item[\rm(3.ii)] \vskip-2pt
  $a\wh e_{r\+j}\ol{a\wh e_{r\+j}}+b\wh e_{r\+j}\ol{b\wh e_{r\+j}}=0$, 
  for $j=1,\cdots,s$.
 \end{itemize}
\end{itemize}
\end{lemma}

\begin{proof}
(2)$\Rightarrow$(1)
For any elements $u(a,b),u'(a,b)\in C_{(a,b)}$,
the inner product
\begin{align*}
  \big \langle u(a,b),u'(a,b)\big \rangle
=\sigma (ua\cdot \overline{u'a})+\sigma(ub\cdot\overline{u'b})
=\sigma\big(u\overline{u'}(a\ol a+b\ol b)\big)=0.
\end{align*}
%then (1) holds.

(1)$\Rightarrow$(2).
The proof is similar to \cite[Lemma II.3]{FL20}.
Suppose $a\bar a+b\bar b=\sum_{x\in G}a_xx$
with a coefficient $a_{x_0}\neq 0$.
Then the inner product of $(a,b)\in C_{a,b}$ and $x_0(a,b)\in C_{a,b}$ is 
$$\textstyle 
\big \langle(a,b),x_0(a,b)\big\rangle
=\sigma\big(\overline{x_0}(a\ol a+b\ol b)\big)
=\sigma\big({x_0}^{-1}\cdot\sum_{x\in G}a_xx\big)
=a_{x_0}\neq 0,
$$
which contradicts to (1). 

(2)$\Leftrightarrow$(3).
Since
\begin{align*}
&a=ae_0+ae_1+\cdots+ae_r+a\wh e_{r\+1}+\cdots+a\wh e_{r\+s}, \\
&b=be_0+be_1+\cdots+be_r+b\wh e_{r\+1}+\cdots+b\wh e_{r\+s},
\end{align*}
then
\begin{align*}
a\ol a+b\ol b
&=(ae_0\ol{ae_0}+be_0\ol{be_0})
  +(ae_1\ol{ae_1}+be_1\ol{be_1})+\cdots+(ae_r\ol{ae_r}+be_r\ol{be_r})\\
&~~~+(a\wh e_{r\+1}\ol{a\wh e_{r\+1}}+b\wh e_{r\+1}\ol{b\wh e_{r\+1}})+\cdots+(a\wh e_{r\+s}\ol{a\wh e_{r\+s}}+b\wh e_{r\+s}\ol{b\wh e_{r\+s}}).
\end{align*}
Hence, (2) and (3) are equivalent.
\end{proof}

\begin{corollary}\label{self-dual}
Let $C_{1,b}$ be a $2$-quasi-$FG$ code of type I. 
The following are equivalent to each other: 
\begin{itemize}
\item[\rm(1)]\vskip-5pt
 $C_{1,b}$ is self-dual; 
\item[\rm(2)] \vskip-3pt
 $b\ol b=-1$;
\item[\rm(3)]\vskip-3pt
the following two hold:
\begin{itemize}
 \item[\rm(3.i)] \vskip-4pt
  $be_i\ol{be_i}=-e_i$, for $i=0,1,\cdots,r$; 
 \item[\rm(3.ii)] \vskip-2pt
  $b\wh e_{r\+j}\ol{b\wh e_{r\+j}}=-\wh e_{r\+j}$, 
  for $j=1,\cdots,s$.
 \end{itemize}
\end{itemize}
\end{corollary}

\begin{proof}
Since $\dim C_{1,b}=n$ (Lemma \ref{l type I}),  it follows from Lemma \ref{orth}.
\end{proof}

\section{Self-dual $2$-quasi-abelian codes of Type I}\label{self-dual type I}

Keep the notation in Section \ref{2-quasi}.
In the following we always denote
\begin{equation}\label{cal D}
\begin{array}{l}
{\cal D}=\big\{C_{1,b}\,\big|\, b\in FG,\, b\ol b=-1\big\},
\end{array}\end{equation}
which is the set of all self-dual $2$-quasi-$FG$ codes of type I, 
see Corollary \ref{self-dual}.
We always assume that $\delta$ is a real number such that
$0\le\delta\le 1-q^{-1}$, and set
\begin{equation}\label{D^le}
\textstyle
 {\cal D}^{\le\delta}=
\big\{C_{1,b}\,\big|\,C_{1,b}\in{\cal D},\, 
  \Delta(C_{1,b})=\frac{{\rm w}(C_{1,b})}{2n}\le\delta\big\}.
\end{equation}

\subsection{Counting $\big|{\cal D}\big|$}

\begin{theorem}\label{counting D} 
Let the notation be as in Eq.\eqref{d hat E}, 
Eq.\eqref{hat E^dag} and Eq.\eqref{dim n-1}.
Then
%The number of the self-dual $2$-quasi-$FG$ codes of Type I is %equal to 
$$
\big|{\cal D}\big|=
\begin{cases}
 \prod_{i=1}^r (q^{n_i}+1)\prod_{j=1}^s (q^{n_{r\+j}}-1), & \mbox{$q$ is even};\\[4pt]
 2\prod_{i=1}^r (q^{n_i}+1)\prod_{j=1}^s (q^{n_{r\+j}}-1), & q\equiv 1~({\rm mod}~4); \\[4pt]
0, & q\equiv -1~({\rm mod}~4).
\end{cases}
$$
\end{theorem}

\begin{proof}
The cardinality $|{\cal D}|$ is equal to the number of 
the choices of $b\in FG$ such that $b\ol b=-1$. Let 
\begin{equation}\label{cal T}
\begin{array}{l}
{\cal T}_i=\{\beta\,|\,\beta\in FGe_i,\;\beta\ol\beta=-e_i\},
   \quad i=0,1,\cdots,r; \\[4pt]
{\cal T}_{r\+j}=\{\beta\,|\,\beta\in FG\wh e_{r\+j},\;\beta\ol\beta=-\wh e_{r\+j}\},
   \quad j=1,\cdots,s.
\end{array} 
\end{equation}
By Corollary \ref{self-dual}, 
$$\textstyle
 \big|{\cal D}\big|=|{\cal T}_0|\cdot\prod_{i=1}^{r}|{\cal T}_i|
 \cdot\prod_{j=1}^{s}|{\cal T}_{r+j}|.
$$
Thus the proof of the theorem will be completed by the following
Lemma \ref{T_0}, Lemma \ref{T_r} and Lemma \ref{T_s}.
\end{proof}

\begin{lemma}\label{T_0}
Let the notation be as in Eq.\eqref{cal T}. Then 
$$\big|{\cal T}_0\big|=
\begin{cases}
 1, & \mbox{$q$ is even};\\[4pt]
 2, & q\equiv 1~({\rm mod}~4); \\[4pt]
 0, & q\equiv -1~({\rm mod}~4).
\end{cases}$$
\end{lemma}

\begin{proof}
Since $FGe_0=Fe_0$ and $\ol\beta=\beta$ for all $\beta\in Fe_0$,
$|{\cal T}_0|$ is equal to the number of the solutions in $F$
of the equation $X^2=-1$.
Note that the unit group $F^\times$ is a cyclic group of order $q-1$.

If $q$ is even, then $-1=1$, $q-1$ is odd, 
hence $|{\cal T}_0|=1$.
 
Assume that $q$ is odd.  
Then $-1$ is the unique element of order $2$ in $F^\times$. 
Hence, $X^2=-1$ has solutions if and only if 
$F^\times$ has an element of order~$4$, i.e., $4\,|\,q-1$. 
Thus, if $4\,|\,q-1$ then $|{\cal T}_0|=2$; 
otherwise $|{\cal T}_0|=0$. 
\end{proof}

\begin{lemma}\label{T_r}
Keep the notation in Eq.\eqref{cal T}. If $1\le i\le r$ then
$\big|{\cal T}_i\big|=q^{n_i}+1$.
\end{lemma}

\begin{proof} 
Since $FGe_i$ is a finite field with $|FG e_i|=q^{2n_i}$, 
and $e_i$ is the identity element of the field,
then the unit group $(F\!Ge_i)^\times$ is cyclic, and
 $|(F\!Ge_i)^\times|=q^{2n_i}-1=(q^{n_i}+1)(q^{n_i}-1)$.
  Note that $\ol e_i=e_i$. 
Hence by Lemma \ref{FG^flat}(4),
the ``bar'' map is a Galois automorphism of $FGe_i$ of order $2$.
So, for any $\beta\in FGe_i$,
$\ol\beta=\beta^{q^{n_i}}$.
Then the equation $\beta\ol \beta=-e_i$ turns into: 
\begin{equation}\label{g e bar}
 \beta^{q^{n_i}+1}=-e_i.
\end{equation}

Case 1: $q$ is even. The equation 
$X^{q^{n_i}+1}=1$ has exactly $q^{n_i}+1$ roots in the finite field $FGe_i$.
The lemma holds.

Case 2: $q$ is odd. Then $2\,|\,(q^{n_i}-1)$, 
the equation $X^{2(q^{n_i}+1)}=1$ has exactly $2(q^{n_i}+1)$ roots.
Hence,  $X^{q^{n_i}+1}=-1$ has exactly $q^{n_i}+1$ roots in the finite field $FGe_i$.
We are done.
\end{proof}

\begin{lemma}\label{T_s}
Keep the notation in Eq.\eqref{cal T}. If $1\le j\le s$ then
$\big|{\cal T}_{r\+j}\big|=q^{n_{r\+j}}-1$.
\end{lemma}

\begin{proof}
% By assumption, 
%$\ol e_{r\+j}\ne e_{r\+j}$, $\wh e_{r\+j}=e_{r\+j}+\ol e_{r\+j}$, 
%and
%$$
%  FG\wh e_{r\+j}=FGe_{r\+j}\oplus FG\ol e_{r\+j}, \quad
% \dim FGe_{r\+j}=\dim FG\ol e_{r\+j}= n_{r\+j}.
%$$
%Further,
%$$
%   FGe_{r\+j} \to FG\ol e_{r\+j},~~ a\mapsto \ol a,
%$$
%is an algebra isomorphism. We can
 Write
$\beta=\beta' + \beta''$ with 
$\beta'\in FGe_{r\+j}$ and $\beta''\in FG\ol e_{r\+j}$. 
Then $\beta\ol\beta=-\wh e_{r\+j}$ for $\beta\in{\cal T}_{r\+j}$ is rewritten as
$$
-e_{r\+j}-\ol e_{r\+j}=(\beta'+\beta'')\ol{(\beta'+\beta'')}
=(\beta'+\beta'')(\ol{\beta'}+\ol{\beta''})
=\beta'\ol{\beta''} + \ol{\beta'}\beta'',
$$
which is equivalent to 
\begin{equation}\label{g e not bar}
\beta'\ol{\beta''}=-e_{r\+j} \quad
(\mbox{equivalently,~ } \ol{\beta'}\beta''=-\ol e_{r\+j}).
\end{equation}
Take any $\beta'\in(FGe_{r\+j})^\times$ 
there is a unique element $\gamma\in(FGe_{r\+j})^\times$ such that
$\beta'\gamma =-e_{r\+j}$; 
and, there is a unique $\beta''\in(FG\ol e_{r\+j})^\times$ 
such that $\ol{\beta''}= \gamma$. 
Note that $\big|(FGe_{r\+j})^\times\big|=q^{n_{r\+j}}-1$.
In conclusion,  
there are exactly $q^{n_{r\+j}}-1$ elements 
$\beta=\beta'+\beta''\in FG(e_{r\+j} +\ol e_{r\+j})$ 
such that $\beta\ol \beta =-e_{r\+j}-\ol e_{r\+j}=-\wh e_{r\+j}$.
\end{proof}

\begin{corollary}\label{exist iff}
The following four are equivalent to each other:

{\rm(1)} The self-dual 
$2$-quasi abelian codes exist.

{\rm(2)}  The self-dual $2$-quasi abelian codes of Type I exist.

{\rm(3)} $-1$ is a square element of $F$.

{\rm(4)} Either $q$ is even, or $4\,|\,q-1$.
\end{corollary}

\begin{proof}
(3)$\Leftrightarrow$(4). 
Both (3) and (4) are equivalent to that the equation $X^2=-1$
has solutions in $F$, see the proof of Lemma \ref{T_0}.

(1)$\Rightarrow$(4). 
Let $C\le(FG)^2$ such that $C=C^\bot$. By Corollary \ref{C orth},
\begin{align*}
& e_0C\oplus e_1C\oplus\cdots\oplus e_rC
 \oplus \wh e_{r\+1}C \oplus\cdots\oplus \wh e_{r\+s}C
=C=C^\bot\\
= & (e_0C)^{\bot_0}\oplus(e_1C)^{\bot_1}\oplus\cdots\oplus(e_rC)^{\bot_r}
 \oplus(\wh e_{r\+1}C)^{\bot_{r\+1}}\oplus\cdots\oplus(\wh e_{r\+s}C)^{\bot_{r\+s}}.
\end{align*}
In particular, $e_0C=(e_0C)^{\bot_0}$, 
which implies that there is a 
 $\beta\in Fe_0$ such that $\beta\ol\beta=-e_0$; so $|{\cal T}_0|>0$.
By Lemma \ref{T_0}, (4) holds.

(4)$\Rightarrow$(2). By Theorem \ref{counting D}, $|{\cal D}|>0$.

(2)$\Rightarrow$(1). Trivial. 
\end{proof}\vskip-25mm

\subsection{Estimating $\big|{\cal D}^{\le\delta}\big|$}

Before going on, we list two lemmas which will be cited later. The following
\begin{equation}\label{entropy}
 h_q(\delta)=\delta\log_q(q-1)-\delta\log_q\delta-(1-\delta)\log_q(1-\delta),
 ~~~~ \delta\in[0,1-q^{-1}], 
\end{equation}
is the so-called {\em $q$-entropy function}, which is increasing and
concave on the interval $[0,1-q^{-1}]$ with $h_q(0)=0$ and $h_q(1-q^{-1})=1$.  
For any subset $S\subseteq (FG)^2$, we denote 
$\textstyle S^{\le\delta}
 =\big\{(a,b)\;\big|\;(a,b)\in S,~\frac{{\rm w}(a,b)}{2n}\le\delta\big\}.
$

\begin{lemma}\label{balanced}
%Assume that $\delta$ be a real number with $0<\delta<1-q^{-1}$.
If $A\le FG$, then $A\times A\le(FG)^2$ and
$\big|(A\times A)^{\le\delta}\big|
 \le q^{h_q(\delta)\cdot \dim (A\times A)}.$
\end{lemma}

\begin{proof}
It follows from \cite[Corollary 3.4 and Corollary 3.5]{FL15}.
\end{proof}

\begin{lemma}\label{q^k-1,q^k+1}
Let $q\ge 2$, $k_1,\cdots, k_m$ be integers.
If $k_i\ge \log_q\!m$ for $i=1,\cdots,m$, then 

{\rm (1)}~ $
(q^{k_1}-1)\cdots(q^{k_m}-1)\ge q^{k_1+\cdots+k_m-2};$

{\rm(2)}~ $(q^{k_1}+1)\cdots(q^{k_m}+1)\le q^{k_1+\cdots+k_m+2}.$
\end{lemma}

\begin{proof}
(1). Please see {\cite[Lemma 2.9]{FL20}}.

(2). Assume that $k_1\le\cdots\le k_m$. 
Since $\log_q m\le k_1$, $\frac{m}{q^{k_1}}\le 1$.
\begin{align*}  
\textstyle\frac{(q^{k_1}+1)\cdots(q^{k_m}+1)}{(q^{k_1})\cdots(q^{k_m})}
\textstyle
  =\left(1+\frac{1}{q^{k_1}}\right)\cdots\left(1+\frac{1}{q^{k_m}}\right)
\textstyle
  \le\left(1+\frac{1}{q^{k_1}}\right)^m
   \le \left(1+\frac{1}{q^{k_1}}\right)^{q^{k_1}}.
\end{align*}
Since the sequence $(1+\frac{1}{k})^k$ for $k=1,2,\cdots$ is increasing and
bounded above by~$3$ and $3<q^2$, we obtain the wanted inequality.  
\end{proof}

\medskip
Recall that
$\wh{E}=\{e_0,e_1,\cdots,e_r, \wh e_{r\+1},\cdots, \wh e_{r\+s}\}$, 
$\wh{E}^\dag=\wh E-\{e_0\}$, 
%\{e_1,\cdots,e_r, \wh e_{r\+1},\cdots, \wh e_{r\+s}\}$, 
 and for $e\in\wh E^\dag$, $n_e=\frac{1}{2}\dim FG e$, 
 see Eq.\eqref{d hat E}, Eq.\eqref{hat E^dag} and Eq.\eqref{dim n-1}.

For any $a\in FG$, we denote: 
\begin{equation}\label{ell_a}
\begin{array}{l}
 {\wh E}_a=\big\{e\,|\,e\in{\wh E},\, ea\ne 0\big\},~~~~ 
 \wh E^\dag_a=\big\{e\,|\,e\in\wh E^\dag,\, ea\ne 0\big\};
\\[6pt]
 L_{a}:=\bigoplus_{e\in{\wh E}^\dag_a}FGe\le FG,~~~~ 
 \ell_a:=\sum_{e\in{\wh E}^\dag_a}n_e ~~
 (\mbox{hence $\dim L_{a}= 2\ell_a$}).
\end{array}
\end{equation}

\begin{lemma}\label{D_a,d}
For $(a,d)\in (FG)^2$,
let $${\cal D}_{(a,d)}=\{C_{1,b}\,|\,C_{1,b}\in{\cal D}, (a,d)\in C_{1,b}\}.$$
If ${\cal D}_{(a,d)}\ne\emptyset$, then
${\wh E}_a={\wh E}_d$ and 
%$$\textstyle
% \prod_{e\in \wh E^\dag-\wh E^\dag_a}(q^{n_e}-1)
% \le \big|{\cal D}_{(a,d)}\big|
%\le 2\prod_{e\in\wh E^\dag-\wh E^\dag_a}(q^{n_e}+1).$$
$$\textstyle
\big|{\cal D}_{(a,d)}\big|
\le 2\prod_{e\in\wh E^\dag-\wh E^\dag_a}(q^{n_e}+1).$$
\end{lemma}

\begin{proof}
Assume that $C_{1,b}\in{\cal D}_{(a,d)}$,
then $C_{(1,b)}\in \cal D$,
i.e., $b\ol{b}=-1$ and
$(a,d)=u(1,b)$ for some $u\in FG$. Then
$a=u$ and $d=ub=ab$.
Assume  $b=\sum_{e\in\wh E}\beta_e$ with $\beta_e\in FGe$.
Then,  for $e\in \wh E$, 
\begin{equation}\label{beta_e}
 \beta_e\ol\beta_e=-e, \qquad ed=ea\beta_e.
\end{equation} 
Hence, $\beta_e\in(FGe)^\times$.
So $ea\ne 0$ iff $ed=ea\beta_e\ne 0$.
That is, $\wh E_a=\wh E_d$.
%And
%$$\textstyle
%  FG a=\bigoplus_{e\in\wh E_a} FGae
% =\bigoplus_{e\in\wh E_a} FGabe
% =\bigoplus_{e\in\wh E_d} FGde
% =FGd.
%$$
%If $ae_0\ne 0$, then $ae_0\in Fe_0$ is invertible, and $de_0=ae_0be_0$; hence
%$$
%((ae_0)^{-1}(de_0))^2=(be_0)^2=be_0\ol{be_0}=-e_0.
%$$

%Next, assume that both the conditions (1) and (2) hold.  
%By (1), $\wh E_a=\wh E_d$ obviously.
%Next, we select $b\in FG$ such that 
%\begin{equation}\label{find b}
 %C_{1,b}\in{\cal D}~\mbox{ (i.e., $b\ol b=-1$)~~~~ and~~}~~ 
 %(a,d)\in C_{1,b} ~\mbox{ (i.e., $d=ab$)}.
%\end{equation} 
Note that there already exist such $\beta_e$'s satisfying Eq.\eqref{beta_e}
because ${\cal D}_{(a,d)}\ne\emptyset$. 
What we are computing is how many choices of them.    

{\it Case 1}: $e=e_0$. We have seen from Lemma \ref{T_0} that
there are at most two choices of $\beta_{e_0}$.

{\it Case 2}:  $e\in \wh E^\dag_a$. There are two subcases. 

{\it Subcase 2.1}: $e=e_i$ for some $1\le i\le r$. 
Then $FGe_i$ is a field; and $ae_i\ne 0\ne de_i$.
So, $de_i=ae_i\beta_{e_i}$ implies that
$\beta_{e_i}=(ae_i)^{-1}(de_i)$ is uniquely determined.

{\it Subcase 2.2}: $e=\wh e_{r\+j}=e_{r\+j}+\ol e_{r\+j}$ for some $1\le j\le s$. 
Then $\beta_{\wh e_{r\+j}}=\beta'+\beta''$ with
$\beta'\in (FGe_{r\+j})^\times$ and $\beta''\in (FG\ol e_{r\+j})^\times$
such that $\beta'\ol{\beta''}=-e_{r\+j}$,  see~Eq.\eqref{g e not bar}.
Since $\wh e_{r\+j}a\ne 0$, $e_{r\+j}a\ne 0$ or $\ol e_{r\+j}a\ne 0$ (or both).
Without loss of generality, assume that $e_{r\+j} a\ne 0$.
In $FGe_{r\+j}$, by Eq.\eqref{beta_e}, $e_{r\+j} d=e_{r\+j} a\beta'$, 
hence $\beta'=(e_{r\+j} a)^{-1}(e_{r\+j} d)$ is uniquely determined.
Consequently, ${\beta''}=\ol{-\beta'^{-1}}$ is also uniquely determined.

Combining the two subcases, we conclude that,
if $e\in{\wh E}^\dag_a$, 
there is a unique $\beta_e\in FGe$ satisfying Eq.\eqref{beta_e}.

{\it Case 3}:  $e\in{\wh E}^\dag-{\wh E}^\dag_a$. 
Then $ae=0=de$, hence the second equality of Eq.\eqref{beta_e} always holds.
Thus, there are $q^{n_e}+1$ choices for $\beta_e$ (see Lemma \ref{T_r}), 
or $q^{n_e}-1$ choices for $\beta_e$ (see Lemma \ref{T_s}).

Summarizing the above three cases, we get the inequalities of the lemma.
\end{proof}

\begin{corollary}\label{|D|<}
Keep the notation as above. 
$$ |{\cal D}_{(a,d)}|\le q^{\frac{n-1}{2}-\ell_a+3}.$$
\end{corollary}
\begin{proof}
By Eq.\eqref{dim n-1},
$\frac{n-1}{2}
=\sum_{e\in{\wh E}^\dag} n_e
=\sum_{e\in{\wh E_a}^\dag}n_e
+\sum_{e\in\wh E^\dag-\wh E^\dag_a}n_e.$
By the above lemma %~\ref{D_a,d} 
and Lemma~\ref{q^k-1,q^k+1}(2),
we have: $\big|{\cal D}_{(a,d)}\big|\le 
2\prod_{e\in{\wh E}^\dag-{\wh E}^\dag_a}(q^{n_e}+1)
\le 2\cdot q^2\cdot q^{\sum_{e\in{\wh E}^\dag-{\wh E}^\dag_a}n_e}
\le q^3\cdot q^{\frac{n-1}{2}-\ell_a}$.
\end{proof}

It is known by~\cite{AKS} that
\begin{equation}\label{mu_q(n)}
 \min\big\{\dim FG e\,\big|\,e\in E-\{e_0\}\big\}=\mu_q(n),
\end{equation}
where $\mu_q(n)$ is the minimal size of non-trivial $q$-cyclotomic cosets
on ${\Bbb Z}_n$. For $e\in\wh E^\dag$, by Lemma \ref{FG^flat}, 
$\dim (FGe)$ is even and $\dim (FGe)\ge\mu_q(n)$.

For an integer $\ell$ with $\mu_q(n)\le 2\ell\le n-1$, we denote
\begin{equation}\label{Omega}
\begin{array}{l}
\Omega_{2\ell}=\big\{J\le F\!G\,\big|\,
 \mbox{$J$ is a direct sum of some of $(F\!G e)$'s for $e\in{\wh E}^\dag$},
   \dim J\!=\!2\ell\big\}. \\
%\tilde\Omega_{2\ell}=\big\{\tilde J=Fe_0+J\le F\!G\,\big|\,J\in\Omega_{2\ell}\big\}
%\quad (\mbox{hence~$\dim \tilde J=2\ell+1$ for $\tilde J\in\tilde\Omega_{2\ell}$}).
\end{array}
\end{equation}

\begin{lemma}\label{tilde Omega}
%$\big|\Omega_{2\ell}\big|=\big|\tilde\Omega_{2\ell}\big|
% \le (m-1)^{2\ell/\mu_q(n)}\le n^{2\ell/\mu_q(n)}$.
 $\big|\Omega_{2\ell}\big|
  \le \big|{\wh E}^\dag\big|^{2\ell/\mu_q(n)}\le n^{2\ell/\mu_q(n)}$.
\end{lemma}

\begin{proof} 
For $J\in\Omega_{2\ell}$, $\dim J=2\ell$ and $J$ is a direct sum of some of 
$(FG e)$'s for $e\in{\wh E}^\dag$. Since $\dim FG e=2n_e\ge \mu_q(n)$, 
 the number of the direct summands in $J$ is at most $2\ell/\mu_q(n)$.
And, since $|{\wh E}^\dag|=r+s\le n-1$, 
the number of the choices of each direct summand of $J$ is at most $r+s$. 
\end{proof}

%For $\tilde J\in\tilde \Omega_{2\ell}$, we denote
%$$
% \tilde J^*=\big\{ a\,\big|\, a\in\tilde J,~\ell_a=\ell\big\}.
%$$xxxxx
For $J\in\Omega_{2\ell}$, 
we denote (where $\ell_a$ and $L_a$ are defined in Eq.\eqref{ell_a})
\begin{equation*}
\begin{array}{l}
 \tilde J=FGe_0+J~~\mbox{(hence $\dim\tilde J=2\ell+1$)};
  \\[5pt]
 \tilde J^*=\big\{ a\,\big|\, a\in\tilde J,~L_a=J
 \mbox{ (equivalently, $\ell_a=\ell$)}\big\}.
\end{array}
\end{equation*}

\begin{lemma}\label{D le} ~
${\cal D}^{\le\delta}
= \bigcup_{\ell=\frac{1}{2}\mu_q(n)}^{(n-1)/2}
   \bigcup_{J\in\Omega_{2\ell}}
   \bigcup_{(a,d)\in (\tilde J^*\times\tilde J^*)^{\le\delta}}{\cal D}_{(a,d)}$.
\end{lemma}

\begin{proof}
If $(a,d)\in (\tilde J^*\times\tilde J^*)^{\le\delta}$, then
$0<\frac{{\rm w}(a,d)}{2n}\le\delta$.
%Assume ${\cal D}_{(a,d)}\ne\emptyset$,
For any $C_{1,b}\in{\cal D}_{(a,d)}$,
we have
$C_{1,b}\in\cal D$ and $\frac{\rm w(C_{1,b})}{2n}\le \delta$,
i.e., $C_{1,b}\in \cal D^{\le\delta}$.
Hence
${\cal D}_{(a,d)}\subseteq {\cal D}^{\le\delta}$.

Let $C_{1,b}\in{\cal D}^{\le\delta}$.
There is  $0\ne (a, d)\in(FG)^2$ such that 
$$(a,d)\in C_{1,b}~~~\mbox{and}~~~0<{\rm w}(a,d)\le 2\delta n.$$
Then
$C_{1,b}\in {\cal D}_{(a,d)}$ and
 ${\wh E}^\dag_a\ne\emptyset$
(otherwise $(a,d)=(\alpha e_0,\alpha' e_0)$ with 
$0\ne \alpha\in F$ and $0\ne \alpha'\in F$, hence ${\rm w}(a,d)=2n>2\delta n$). 
By Lemma \ref{D_a,d},
$\wh E^\dag_d=\wh E^\dag_a$.
Then $(a,d)\in \tilde J^*\times\tilde J^*$ with $J=L_a\in\Omega_{2\ell_a}.$
%So, 
%$$ 
% \mu_q(n)\le\ell_a\le (n-1)/2, ~~~
% L_a\in\Omega_{2\ell_a}\cup\tilde\Omega_{2\ell_a}, ~~~
% \mbox{and} ~~ C_{1,b}\in{\cal D}_{(a,d)}. 
%$$
%If $L_a\in\tilde\Omega_{2\ell_a}$ then
%$(a,d)\in (L_a^*\times L_a^*)^{\le\delta}$.
%Otherwise, $L_a\in\Omega_{2\ell_a}$, 
%$\tilde L_a=Fe_0+L_a\in\tilde\Omega_{2\ell_a}$
%and $(a,d)\in (\tilde L_a^*\times\tilde L_a^*)^{\le\delta}$.
Thus, the left hand side of the equality in the lemma
 is contained in the right hand side.  
\end{proof}

\begin{lemma}\label{in tilde Omega}
Let $\frac{1}{2}\mu_q(n)\le\ell\le\frac{n-1}{2}$ and  $J\in\Omega_{2\ell}$.
Then
$$\textstyle
\big|\bigcup_{(a,d)\in (\tilde J^*\times\tilde J^*)^{\le\delta}}{\cal D}_{(a,d)}\big|
\le q^{\frac{n-1}{2}+4\ell\left(h_q(\delta)-\frac{1}{4}\right)+2h_q(\delta)+3}.
$$
\end{lemma}

\begin{proof} 
Since %$J\in\Omega_{2\ell}$,
 $\dim(\tilde J\times\tilde J)= 2(2\ell+1)=4\ell+2$, 
 by Lemma \ref{balanced} we have 
$$ \big|(\tilde J\times\tilde J)^{\le\delta}\big|
\le q^{(4\ell+2)h_q(\delta)}.$$
For $(a,d)\in (\tilde J^*\times \tilde J^*)^{\le\delta}$, 
%if ${\cal D}_{(a,d)}\ne\emptyset$, then 
$\ell_a=\ell$ and 
$\big|{\cal D}_{(a,d)}\big|\le q^{\frac{n-1}{2}-\ell+3}$
(see Corollary~\ref{|D|<}).
So
\begin{align*}
&\textstyle
\big|\bigcup_{(a,d)\in (\tilde J^*\times \tilde J^*)^{\le\delta}}{\cal D}_{(a,d)}\big|
\le\sum_{(a,d)\in (\tilde J^*\times \tilde J^*)^{\le\delta}}
   \big|{\cal D}_{(a,d)}\big|\\[4pt]
&\textstyle \le \sum_{(a,d)\in (\tilde J^*\times \tilde J^*)^{\le\delta}}
   q^{\frac{n-1}{2}-\ell+3}
\le \sum_{(a,d)\in (\tilde J\times \tilde J)^{\le\delta}}
     q^{\frac{n-1}{2}-\ell+3}\\[4pt]
&=\big|(\tilde J\times \tilde J)^{\le\delta}\big|
     \cdot q^{\frac{n-1}{2}-\ell+3}
\le q^{(4\ell+2)h_q(\delta)}\cdot q^{\frac{n-1}{2}-\ell+3}\\[4pt]
&=q^{\frac{n-1}{2}+4\ell\left(h_q(\delta)-\frac{1}{4}\right)+2h_q(\delta)+3}.
\end{align*}
We are done.
\end{proof}

\begin{theorem}\label{D^le <}
Assume that 
$\frac{1}{4}-h_q(\delta)-\frac{\log_q n}{2\mu_q(n)}>0$. 
Then
$$
\big|{\cal D}^{\le\delta}\big|\le 
q^{\frac{n-1}{2}-2\mu_q(n) 
     \left(\frac{1}{4}-h_q(\delta)-\frac{\log_q n}{\mu_q(n)}\right) +4}.
$$
\end{theorem}

\begin{proof} Denote $\mu=\mu_q(n)$.
By Lemma \ref{D le},  Lemma \ref{in tilde Omega} and Lemma \ref{tilde Omega}, 
\begin{align*}
|{\cal D}^{\le\delta}|
&\textstyle
 \le\sum_{\ell=\frac{1}{2}\mu}^{\frac{n-1}{2}}\sum_{ J\in \Omega_{2\ell}} 
 q^{\frac{n-1}{2}+4\ell\left(h_q(\delta)-\frac{1}{4}\right)+2h_q(\delta)+3}\\
&\textstyle
 \le\sum_{\ell=\frac{1}{2}\mu}^{\frac{n-1}{2}}n^{2\ell/\mu}\cdot 
 q^{\frac{n-1}{2}+4\ell\left(h_q(\delta)-\frac{1}{4}\right)+2h_q(\delta)+3}\\
&\textstyle
 =\sum_{\ell=\frac{1}{2}\mu}^{\frac{n-1}{2}} q^{\frac{n-1}{2}
  -4\ell\left(\frac{1}{4}-h_q(\delta)-\frac{\log_q n}{2\mu}\right)
  +2h_q(\delta)+3}.
\end{align*}
Since $2\ell\ge \mu$ and $\frac{1}{4}-h_q(\delta)-\frac{\log_q n}{2\mu}>0$
(hence $2h_q(\delta)\le 1$),
\begin{align*}
|{\cal D}^{\le\delta}|
&\textstyle
 \le \sum_{\ell=\frac{1}{2}\mu}^{\frac{n-1}{2}}
  q^{\frac{n-1}{2}-2\mu \left(\frac{1}{4}-h_q(\delta)-\frac{\log_q n}{2\mu}\right)
  +1+3}\\
&\le n\cdot
q^{\frac{n-1}{2}-2\mu \left(\frac{1}{4}-h_q(\delta)-\frac{\log_q n}{2\mu}\right)
  +4}\\
&=q^{\frac{n-1}{2}-2\mu \left(\frac{1}{4}-h_q(\delta)-\frac{\log_q n}{\mu}\right)
  +4}. 
\end{align*}
We are done.
\end{proof}

\subsection{Asymptotic goodness}

\begin{theorem} \label{D^le/D}
Assume that $q$ is even or $4\,|\,(q-1)$, and 
$\frac{1}{4}-h_q(\delta)-\frac{\log_q n}{2\mu_q(n)}>0$. 
Then
$$
\left.|{\cal D}^{\le\delta}|\right/|{\cal D}|\le 
q^{-2\mu_q(n) \left(\frac{1}{4}-h_q(\delta)-\frac{\log_q n}{\mu_q(n)}\right)+6}.
$$
\end{theorem}

\begin{proof} 
By Theorem \ref{cal D} and Lemma \ref{q^k-1,q^k+1} ,
\begin{equation*}
\textstyle
|{\cal D}|\ge \prod_{e\in{\wh E}^\dag}(q^{n_e}-1)
\ge q^{\sum_{e\in{\wh E}^\dag}n_e -2}=q^{\frac{n-1}{2}-2}.
\end{equation*}
By Theorem \ref{D^le <},  
\begin{align*}
 |{\cal D}^{\le\delta}|\Big/|{\cal D}|\;
 & \le\;  q^{\frac{n-1}{2}-2\mu_q(n) 
  \left(\frac{1}{4}-h_q(\delta)-\frac{\log_q n}{\mu_q(n)}\right) +4}
\Big/ q^{\frac{n-1}{2}-2}\\
 &=\; q^{-2\mu_q(n) \left(\frac{1}{4}-h_q(\delta)-\frac{\log_q n}{\mu_q(n)}\right)+6}.
\end{align*}
We are done.
\end{proof}

\begin{lemma}[{\cite[Lemma 2.6]{FL20}}]\label{n_i}
There are infinitely many positive odd integers $n_1,n_2,\cdots$ coprime to $q$
such that $\lim\limits_{i\to\infty}\frac{\log_q n_i}{\mu_q(n_i)}=0$;
in particular, $\mu_q(n_i)\to\infty$. 
\end{lemma}

\begin{remark}\label{r n_i}\rm
By \cite[Lemma 2.6]{FL20}, there are infinitely many primes
$p_1,p_2,\cdots$
such that $\lim\limits_{i\to\infty}\frac{\log_q p_i}{\mu_q(p_i)}=0$.
We just remark that the $n_i$'s in the lemma are not necessarily primes
(which implies that %we can take non-cyclic 
the abelian groups $G_i$ of order $n_i$ are not necessarily cyclic). 
Because: by \cite{AKS}, for any integer $m>1$ coprime to $q$, 
$$
 \mu_q(m)=\min\big\{\,\mu_q(p)\,\big|\; 
 \mbox{$p$ runs over the prime divisors of $m$}\big\};
$$
therefore, if we take, for example, $n_i=p_i^2$, we still have
 $\lim\limits_{i\to\infty}\frac{\log_q n_i}{\mu_q(n_i)}=0$.
\end{remark}

\begin{theorem}\label{D asymptotically good}
Let $n_1,n_2,\cdots$ be as in Lemma \ref{n_i}. 
Let $G_i$ be any abelian group of order $n_i$ for $i=1,2,\cdots$.
If $q$ is even or $q\equiv 1~({\rm mod}~4)$, then
there exist self-dual $2$-quasi-$FG_i$ codes $C_i$ of type I, $i=1,2,\cdots$, 
such that the code sequence $C_1,C_2,\cdots$ is asymptotically good.
\end{theorem}

\begin{proof}
Let ${\cal D}_i$ be the set of all self-dual $2$-quasi-$FG_i$ codes of type I, 
let ${\cal D}_i^{\le\delta}$ be as in Eq.\eqref{D^le}. 
By the property of $h_q(\delta)$ in Eq.\eqref{entropy},
we can take $\delta\in(0,1-q^{-1})$ 
satisfying that $0<h_q(\delta)<\frac{1}{4}$.
Since $\lim\limits_{i\to\infty}\frac{\log_q n_i}{\mu_q(n_i)}=0$,
we can further assume that 
$\frac{1}{4}-h_q(\delta)-\frac{\log_q n_i}{\mu_q(n_i)}>\varepsilon>0$
for a positive real number~$\varepsilon$.
Since $\mu_q(n_i)\to\infty$, by~Theorem \ref{D^le/D} we have
$$
\lim\limits_{i\to\infty}\left.|{\cal D}_i^{\le\delta}|\right/ |{\cal D}_i|
\le \lim\limits_{i\to\infty} q^{-2\mu_q(n_i) 
 \left(\frac{1}{4}-h_q(\delta)-\frac{\log_q n_i}{\mu_q(n_i)}\right)+6}
 =0.
$$
Thus we can take $C_i\in{\cal D}_i-{\cal D}_i^{\le\delta}$ for $i=1,2,\cdots$.
The self-dual $2$-quasi-$FG_i$ codes $C_i$ of type I satisfy the following: 
\begin{itemize}
\item \vskip-5pt
the length $2n_i$ of $C_i$ is going to $\infty$;
\item\vskip-5pt
the rate $R(C_i)=\frac{1}{2}$ for $i=1,2,\cdots$;
\item\vskip-5pt
the relative minimum distance $\Delta(C_i)>\delta$  for $i=1,2,\cdots$.  
\end{itemize}
\vskip-5pt
That is, the sequence of codes $C_1,C_2,\cdots$ is asymptotically good.
\end{proof}

\section{Self-orthogonal $2$-quasi-abelian codes}\label{self-orth type I^dag}

Keep the notation in Sections 2-4.

The existence of self-dual $2$-quasi-$FG$ codes 
(i.e., the self-orthogonal $2$-quasi-$FG$ codes of dimension $n$)
is conditional, see Corollary \ref{exist iff}. However, in the following 
we show that the self-orthogonal $2$-quasi-$FG$ codes 
of dimension $n-1$ always exist, 
and they are asymptotically good.

As exhibited in the proof of Theorem \ref{counting D}, 
the existence of self-dual $2$-quasi-$FG$ codes depends only on
the computation in the $e_0$-component (Lemma~\ref{T_0}).
Similarly to Eq.\eqref{hat E^dag},
%Recalling Eq.\eqref{e hat E} and Eq.\eqref{FG=} and 
by removing the $e_0$-component we set  
\begin{equation*}
\begin{array}{l}
1^{\!\dag}=1-e_0,~~{\rm i.e.,} ~~
1^{\!\dag}=e_1+\cdots+e_r+\wh e_{r\+1}+\cdots+\wh e_{r\+s};
\\[5pt]
FG^\dag=FG\cdot 1^{\!\dag}
  =\bigoplus_{e\in \wh E^\dag}FGe.
\end{array}
\end{equation*}
Then the $e_0$-component of any $b^\dag\in FG^{\dag}$ vanishes, 
i.e., 
$e_0b^\dag=0$ and $b^\dag=1^{\!\dag} b^\dag$. 

\begin{lemma}\label{self-orth n-1}
For any $b^\dag\in FG^\dag$ with $b^{\dag}\ol{b^{\dag}}=-1^{\!\dag}$,
we have
$$C_{1^{\!\dag}\!, b^{\dag}}=\{(u1^{\!\dag},\,ub^{\dag})\,|\,u\in FG\}$$ 
is a self-orthogonal $2$-quasi-$FG$ code of dimension $n-1$. 
\end{lemma}

\begin{proof}
%Since
%$b^{\dag}=b^{\dag}e_1+\cdots+b^{\dag}e_r
% +b^{\dag}\wh e_{r\+1}+\cdots+b^{\dag}\wh e_{r\+s},
%$
%$b^{\dag}\ol{b^{\dag}}=-1^{\!\dag}$ if and only if the following two holds:
%$$\begin{array}{ll}
% b^{\dag}e_i\ol{{b^{\dag}}e_i}=-e_i, & i=1,\cdots,r;\\[4pt]
% b^{\dag}\wh e_{r\+j}\ol{{b^{\dag}}\wh e_{r\+j}}=-\wh e_{r\+j}, 
% & j=1,\cdots,s.
%\end{array}$$
%
By Lemma \ref{orth}, $C_{1^{\!\dag}\!, b^{\dag}}$ is self-orthogonal.
Since ${\rm Ann}_{FG}(1^{\!\dag}\!, b^{\dag})=Fe_0$,
by Lemma~\ref{dimension of Ca,b}, $\dim C_{1^{\!\dag}\!, b^{\dag}}=n-1$.
\end{proof}

Similar to Eq.\eqref{cal D}, we set
\begin{equation}\label{cal D^dag}
\begin{array}{l}
{\cal D}^{\dag}=\big\{C_{1^{\!\dag}\!, b^{\dag}}\,\big|\,
b^{\dag}\in FG^\dag, \, b^{\dag}\ol{b^{\dag}}=-1^{\!\dag}\big\}. 
\end{array}\end{equation}
And, similar to Eq.\eqref{D^le},
we assume that $0<\delta<1-q^{-1}$,  and denote
\begin{equation}\label{dag D^le}
\textstyle
 ({\cal D}^\dag)^{\le\delta}=
\big\{C_{1^{\!\dag}\!,b^\dag}\,\big|\,C_{1^{\!\dag}\!,b^\dag}\in{\cal D}^\dag,\, 
  \frac{ {\rm w} ( C_{ 1^{\!\dag}\!,b^\dag } ) } {2n} \le \delta\big\}.
\end{equation}

\begin{theorem}\label{counting D^dag}  
$\big|{\cal D}^\dag\big|=
\prod_{i=1}^r (q^{n_i}+1)\prod_{j=1}^s (q^{n_{r\+j}}-1)$.
\end{theorem}

\begin{proof}
For $b^{\dag}\in FG^\dag$ such that $b^{\dag}\ol{b^{\dag}}=-1^{\!\dag}$,  
the $e_0$-components of both $1^{\!\dag}$ and $b^\dag$ vanish.
Thus $\big|{\cal D}^\dag\big|=\prod_{k=1}^{r+s}|{\cal T}_k|$,
and the theorem follows from Lemma \ref{T_r} and Lemma \ref{T_s}.
\end{proof}

\begin{remark}\rm
If $C_{1,b}\in{\cal D}$ and $b^\dag=b-be_0$, then
$b^\dag\ol{b^\dag}=-1^{\!\dag}$ and
$C_{1^{\!\dag}\!,b^\dag}\in{\cal D}^{\dag}$. 
In other words, $C_{1^{\!\dag}\!,b^\dag}$ can be obtained by
removing the $e_0$-component of $C_{1,b}$.
However, it may happen that ${\cal D}=\emptyset$, 
see Theorem \ref{counting D}. 
As a comparison, it is always true that ${\cal D}^\dag\ne\emptyset$.
We call $C_{1^{\!\dag}\!,b^\dag}\in{\cal D}^{\dag}$ a
self-orthogonal $2$-quasi-$FG$ code of {\em type $I^\dag$}.
\end{remark}

For $(a,d)\in (FG)^2$,  
let ${\cal D}^\dag_{(a,d)}=
\{C_{1^{\!\dag}\!,b^\dag}\,|\,C_{1^{\!\dag}\!,b^\dag}\in{\cal D}^\dag, 
 (a,d)\in C_{1^{\!\dag}\!,b^\dag}\}$.
And, similar to Lemma \ref{D_a,d}, we have

\begin{lemma}\label{D^dag_a,d}
Let the notation be as above. 
If ${\cal D}^\dag_{(a,d)}\ne\emptyset$, then $ae_0=0=de_0$, 
${\wh E}^\dag_a={\wh E}^\dag_d$ and 
%$$\textstyle
% \prod_{e\in \wh E^\dag-\wh E^\dag_a}(q^{n_e}-1)
% \le \big|{\cal D}^\dag_{(a,d)}\big|
%\le \prod_{e\in\wh E^\dag-\wh E^\dag_a}(q^{n_e}+1).$$
$$\textstyle
 \big|{\cal D}^\dag_{(a,d)}\big|
\le \prod_{e\in\wh E^\dag-\wh E^\dag_a}(q^{n_e}+1).$$

\end{lemma}

\begin{proof} If $C_{1^{\!\dag}\!,b^\dag}\in{\cal D}^\dag$ and
$(a,d)\in C_{1^{\!\dag}\!,b^\dag}$, then there is a $u\in FG$ such that
$a=u1^{\!\dag}$ and $d=ub^\dag$, hence
$ae_0=u1^{\!\dag}e_0=0$,  $de_0=ub^\dag e_0=0$, 
and ${\wh E}^\dag_a={\wh E}^\dag_d$. 
Similar to the proof of Lemma \ref{D_a,d}, 
we can complete the proof
%get the inequality for $\big|{\cal D}^\dag_{(a,d)}\big|$
(just note that Case 1 of the proof of Lemma \ref{D_a,d}
 is no longer present here). 
\end{proof}

Similar to Corollary \ref{|D|<} and Lemma \ref{D le}, we have

\begin{corollary}\label{|D^dag|<}~
$|{\cal D}^\dag_{(a,d)}|\le q^{\frac{n-1}{2}-\ell_a+2}$.
\end{corollary}

\begin{lemma}\label{D^dag le}~
$({\cal D^\dag})^{\le\delta}
= \bigcup_{\ell=\frac{1}{2}\mu_q(n)}^{(n-1)/2}
 \bigcup_{J\in\Omega_{2\ell}}
\bigcup_{(a,d)\in (J^*\times J^*)^{\le\delta}}{\cal D}^\dag_{(a,d)}$,
where
$ J^*=\big\{ a\,\big|\, a\in J,~L_a=J\big\}.$
\end{lemma}

And Lemma \ref{in tilde Omega} is revised as follows.

\begin{lemma}\label{in Omega}
Let $\frac{1}{2}\mu_q(n)\le\ell\le\frac{n-1}{2}$ and  $J\in\Omega_{2\ell}$.
Then
$$\textstyle
\big|\bigcup_{(a,d)\in (J^*\times J^*)^{\le\delta}}{\cal D}^\dag_{(a,d)}\big|
\le q^{\frac{n-1}{2}+4\ell\left(h_q(\delta)-\frac{1}{4}\right)+2}.
$$
\end{lemma}

\begin{proof} 
It is similar to the computation in the proof of Lemma \ref{in tilde Omega} 
(but $\tilde J$ is replaced by $J$, etc.): 
%For $(a,d)\in (\tilde J^*\times \tilde J^*)^{\le\delta}$, 
%if $\big|{\cal D}_{(a,d)}\big|>0$, then  
%$\big|{\cal D}_{(a,d)}\big|\le q^{\frac{n-1}{2}-\ell+3}$,
%see~Corollary \ref{|D|<}.
%Thus,
\begin{align*}
&\textstyle
\big|\bigcup_{(a,d)\in (J^*\times J^*)^{\le\delta}}{\cal D}^\dag_{(a,d)}\big|
\le\sum_{(a,d)\in (J^*\times J^*)^{\le\delta}}
   \big|{\cal D}^\dag_{(a,d)}\big|\\[4pt]
&\textstyle \le \sum_{(a,d)\in (J^*\times J^*)^{\le\delta}}
   q^{\frac{n-1}{2}-\ell+2}
\le \sum_{(a,d)\in (J\times J)^{\le\delta}}
     q^{\frac{n-1}{2}-\ell+2}\\[4pt]
&=\big|(J\times J)^{\le\delta}\big|\cdot q^{\frac{n-1}{2}-\ell+2}
\le q^{4\ell h_q(\delta)}\cdot q^{\frac{n-1}{2}-\ell+2}\\[4pt]
&=q^{\frac{n-1}{2}+4\ell\left(h_q(\delta)-\frac{1}{4}\right)+2},
\end{align*}
where the last second line follows from Lemma~\ref{balanced} 
 and $\dim(J\times J)= 4\ell$.
\end{proof}

Then, similar to Theorem \ref{D^le <} and Theorem \ref{D^le/D},
 we can get the following results. 

\begin{theorem}\label{D^dag^le <}
Assume that 
$\frac{1}{4}-h_q(\delta)-\frac{\log_q n}{2\mu_q(n)}>0$. 
Then
$$
\big|({\cal D}^\dag)^{\le\delta}\big|\le 
q^{\frac{n-1}{2}-2\mu_q(n) 
     \left(\frac{1}{4}-h_q(\delta)-\frac{\log_q n}{\mu_q(n)}\right) +2}.
$$
\end{theorem}

%\begin{proof} Denote $\mu=\mu_q(n)$.
%By Lemma \ref{D^dag le},  Lemma \ref{in Omega} and Lemma \ref{tilde Omega}, 
%\begin{align*}
%\big|({\cal D}^\dag)^{\le\delta}\big|
%&\textstyle
%\le\sum_{\ell=\frac{1}{2}\mu}^{\frac{n-1}{2}}\sum_{J\in \Omega_{2\ell}} 
%q^{\frac{n-1}{2}+4\ell\left(h_q(\delta)-\frac{1}{4}\right)+2}\\
%&\textstyle
%\le\sum_{\ell=\frac{1}{2}\mu}^{\frac{n-1}{2}}n^{2\ell/\mu}\cdot 
%q^{\frac{n-1}{2}+4\ell\left(h_q(\delta)-\frac{1}{4}\right)+2}\\
%&\textstyle
%=\sum_{\ell=\frac{1}{2}\mu}^{\frac{n-1}{2}}
%q^{\frac{n-1}{2}-4\ell\left(\frac{1}{4}-h_q(\delta)-\frac{\log_q n}{2\mu}\right)
%  +2}\\
%&\textstyle
%\le \sum_{\ell=\frac{1}{2}\mu}^{\frac{n-1}{2}}
%q^{\frac{n-1}{2}-2\mu \left(\frac{1}{4}-h_q(\delta)-\frac{\log_q n}{2\mu}\right)
%  +2}\\
%&\le n\cdot
%q^{\frac{n-1}{2}-2\mu \left(\frac{1}{4}-h_q(\delta)-\frac{\log_q n}{2\mu}\right)
%  +2}\\
%&=q^{\frac{n-1}{2}-2\mu \left(\frac{1}{4}-h_q(\delta)-\frac{\log_q n}{\mu}\right)
%  +2}. 
%\end{align*}
%We are done.
%\end{proof}

\begin{theorem} \label{D^dag^le/D^dag}
Assume that 
$\frac{1}{4}-h_q(\delta)-\frac{\log_q n}{2\mu_q(n)}>0$. 
Then
$$
\left.|({\cal D}^\dag)^{\le\delta}|\right/|{\cal D}^\dag|\le 
q^{-2\mu_q(n) \left(\frac{1}{4}-h_q(\delta)-\frac{\log_q n}{\mu_q(n)}\right)+4}.
$$
\end{theorem}

%\begin{proof} 
%By Theorem \ref{counting D^dag} and Lemma \ref{q^k-1,q^k+1},
%\begin{equation*}
%\textstyle
%|{\cal D}^\dag|\ge \prod_{e\in{\wh E}^\dag}(q^{n_e}-1)
%\ge q^{\sum_{e\in{\wh E}^\dag}n_e -2}=q^{\frac{n-1}{2}-2}.
%\end{equation*}
%By Theorem \ref{D^le <},  
%\begin{align*}
%\left.|({\cal D}^\dag)^{\le\delta}|\right/|{\cal D}^\dag|
% & \le\;  q^{\frac{n-1}{2}-2\mu_q(n) 
%  \left(\frac{1}{4}-h_q(\delta)-\frac{\log_q n}{\mu_q(n)}\right) +2}
%\Big/ q^{\frac{n-1}{2}-2}\\
% &=\; q^{-2\mu_q(n) \left(\frac{1}{4}-h_q(\delta)-\frac{\log_q n}{\mu_q(n)}\right)+4}.
%\end{align*}
%We are done.
%\end{proof}

\begin{theorem}\label{D^dag asymptotically good}
Let $n_1,n_2,\cdots$ be odd positive integers coprime to $q$ such that
$\lim\limits_{i\to\infty}\frac{\log_q n_i}{\mu_q(n_i)}=0$. 
Let $G_i$ be any abelian group of order $n_i$ for $i=1,2,\cdots$.
Then for $i=1,2,\cdots$ there exist 
self-orthogonal $2$-quasi-$FG_i$ codes $C_i$ of type $I^\dag$
(hence $\dim C_i=n_i-1$) 
such that the code sequence $C_1,C_2,\cdots$ is asymptotically good.
\end{theorem}

\begin{proof}
Let ${\cal D}^\dag_i$ be the set of all 
self-orthogonal $2$-quasi-$FG_i$ codes of type $I^\dag$.
Take a real number $\delta\in(0,1-q^{-1})$ satisfying that $0<h_q(\delta)<\frac{1}{4}$.
We have
$$
\lim\limits_{i\to\infty}\left.|({\cal D}^\dag_i)^{\le\delta}|\right/ |{\cal D}^\dag_i|
\le \lim\limits_{i\to\infty} q^{-2\mu_q(n_i) 
 \left(\frac{1}{4}-h_q(\delta)-\frac{\log_q n_i}{\mu_q(n_i)}\right)+4}
 =0.
$$
Thus we can take $C_i\in{\cal D}^\dag_i-({\cal D}^\dag_i)^{\le\delta}$ 
for $i=1,2,\cdots$.
The self-orthogonal $2$-quasi-$FG_i$ codes $C_i$ of type $I^\dag$ 
satisfy the following: 
\begin{itemize}
\item \vskip-5pt
the length $2n_i$ of $C_i$ is going to $\infty$;
\item\vskip-5pt
the rate $\lim\limits_{i\to\infty}R(C_i)=\frac{1}{2}$;
\item\vskip-5pt
the relative minimum distance $\Delta(C_i)>\delta$  for $i=1,2,\cdots$.  
\end{itemize}
\vskip-5pt
That is, the sequence of codes $C_1,C_2,\cdots$ is asymptotically good.
\end{proof}

\section{Conclusion}
Let $G$ be any abelian group of odd order $n$ coprime to the 
cardinality $q=|F|$. 
We introduced the self-dual 2-quasi-$FG$ codes of type I, and 
obtained the exact number of such codes.
As a consequence, the self-dual 2-quasi-$FG$ codes exist 
if and only if $-1$ is a square in $F$.
We further estimated the number of the self-dual 2-quasi-$FG$ 
codes of type I with small relative minimum distances.
With the two numerical results we proved that
 the self-dual 2-quasi-$FG$ codes are asymptotically good
 provided $-1$ is a square in $F$. 
  
Moreover, we showed that the self-orthogonal 2-quasi-$FG$ codes
of dimension $n-1$ always exist. 
And, by the same method (with small revisions),
we got two similar numerical results on such codes, 
and showed that 
the self-orthogonal 2-quasi-$FG$ codes are asymptotically good.

%\section*{Acknowledgements}
%NSFC for the supports through Grant 11271005.

\end{document}